\documentclass[letterpaper, DIV=12]{scrartcl}

\usepackage{amssymb, amsthm, mathtools,bbm}
\usepackage[square,sort,comma,numbers]{natbib}
\usepackage[none]{hyphenat}
\newtheorem{theorem}{Theorem}

\newtheorem{lemma}[theorem]{Lemma}
\newtheorem{definition}[theorem]{Definition}
\newtheorem{proposition}[theorem]{Proposition}

\numberwithin{equation}{section}
\numberwithin{theorem}{section}

\newcommand{\cc}{{\mathbb{C}}}

\newcommand{\nn}{{\mathbb{N}}}

\newcommand{\ee}{{\mathbb{E}\,}}

\newcommand{\pp}{{\mathbb{P}}}

\newcommand{\tr}{{\operatorname{Tr}\,}}

\newcommand{\beq}[1]{\begin{equation} \label{#1}}
\newcommand{\eeq}{\end{equation}}

\newcommand{\arcosh}{\operatorname{arcosh}}

\begin{document}
	\addtokomafont{author}{\raggedright}

	\title{ \raggedright The Quantum Random Energy Model\\ as a Limit of p-Spin Interactions}
	
     \author{\hspace{-.075in} Chokri Manai and Simone Warzel}
     \date{\vspace{-.3in}}
	
	\maketitle

	%\begin{abstract}
	\minisec{Abstract}
		We consider the free energy of a mean-field quantum spin glass described by a $ p $-spin interaction and a transversal magnetic field. 
		Recent rigorous results for the case $ p= \infty $, i.e.\ the quantum random energy model (QREM), are reviewed.
		We show that the free energy of the $ p $-spin model converges in a joint thermodynamic and $ p \to \infty $ limit to the free energy
		of the QREM.
	%\end{abstract}

	\section{Introduction}
	A prominent class of classical mean-field spin glass models are $ p $-spin interactions defined on $ N $  Ising-type spins 
	\begin{equation}\label{eq:spins}
	\pmb{\sigma} = (\sigma_1, \dots , \sigma_N) \in \{ -1,1 \}^N =:  \mathcal{Q}_N . 
	\end{equation}
	For fixed $ p \in [1,\infty) $ the interaction energy of these spins is random and given by 
	\begin{equation}\label{eq:U}
	U_p(\pmb{\pmb{\sigma}}) = \frac{1}{N^{\frac{p-1}{2}}} \sum_{j_1, \dots, j_p =1}^N g_{j_1,\dots,j_p}{\sigma}_{j_1} \cdots {\sigma}_{j_p}
	\end{equation}
	in terms of an array  $ g_{\pmb{j}} := g_{j_1,\dots,j_p} $  of independent and identically distributed (i.i.d.), centered Gaussian random variable with variance one. 
	The process $ U_p(\pmb{\sigma}) $, $\pmb{\sigma} \in \mathcal{Q}_N $ is then Gaussian as well and uniquely characterized by its mean and covariance function,
	\begin{equation}\label{eq:spinp}
	\mathbb{E}\left[U_p(\pmb{\sigma}) \right] = 0 , \qquad  \mathbb{E}\left[U_p(\pmb{\sigma}) U_p(\pmb{\sigma}') \right] = N \left( N^{-1} \sum_{j=1}^N \sigma_j \sigma_j' \right)^p =: N \,  \xi(\pmb{\sigma},\pmb{\sigma}')^p  . 
	\end{equation}
	The special case $ p = 2 $ corresponds to the Sherrington-Kirkpatrick model, and in the limit $ p \to \infty $ we obtain 
	Derrida's random energy model (REM) \cite{Der81}. 
	In the latter case, the correlations vanish and the variables $ U_\infty(\pmb{\sigma}) $ form an i.i.d.\ Gaussian process on the hypercube $  \mathcal{Q}_N $.

	There is a wealth of results both in the physics as well as mathematics literature concerning properties of the Gibbs measure of these classical mean-field spin glasses. Most celebrated is a closed form expression for the free energy derived by Parisi \cite{Par80} and later proven by Talagrand and Panchenko \cite{Tal06,Pan14}. This formula reflects the fact that at low temperatures the Gibbs measure fractures into many inequivalent pure states. A key quantity in this area is the distribution of the overlap $ \xi(\pmb{\sigma},\pmb{\sigma}') $ of independent copies or replicas of spins  $ \pmb{\sigma},\pmb{\sigma}'  $. We refer the mathematically interested reader to the monographs \cite{Bov06,Tal11,Pan13} and references therein.

	Despite its popularity in physics (cf.~\cite{SIC12,BFKSZ13} and refs. therein),  much less is rigorously  established if one incorporates quantum effects in the form of a transversal magnetic field.  In the quantum case, one views the spins configurations~\eqref{eq:spins} as the $ z $-components of $ N $ spin-$1/2 $ quantum spins and the energy~\eqref{eq:U} is lifted to the corresponding Hilbert space $  \otimes_{j=1}^N \cc^2 \equiv  \ell^2( \mathcal{Q}_N) $ as a diagonal matrix $ U_p $.  The random Hamiltonian of the quantum p-spin model with transversal magnetic field of strength $\Gamma \geq 0$  is
	\begin{equation}\label{eq:hamiltonp}
	H_p = U_p + \Gamma\,  T , 
	\end{equation} 
	where $ \left( T\psi\right)(\pmb{\sigma}) := -  \sum_{j=1}^N \psi(  \sigma_1, \dots , - \sigma_j  , \dots , \sigma_N  ) $ coincides  with the action of the negative sum of $ x $-components of the Pauli matrices in the $ z $-basis. 
	In this paper, we are concerned with the corresponding quantum free energy or pressure at inverse temperature $\beta \in [0,\infty)$ 
	\begin{equation}\label{eq:F}
	\Phi_N^{p}(\beta,\Gamma) := \frac{1}{N} \ln Z_N^{p}(\beta,\Gamma)
	\end{equation}
	which derives from the partition function $  Z_N^{p}(\beta,\Gamma) = 2^{-N} \tr e^{-\beta H_p} $. The case \mbox{$ p = \infty $} corresponds to the 
	pressure of the quantum random energy model (QREM), and we will write \linebreak $  \Phi^\text{QREM}(\beta,\Gamma) :=  \Phi^{\infty}(\beta,\Gamma) $.
	
	\section{Rigorous results on the free energy} 
	
	A basic property of the free energy of spin glasses is its self-averaging, i.e. the fact that  in the thermodynamic limit $ N \to \infty $ this quantity agrees almost surely with its average. For p-spin interactions even more general than~\eqref{eq:U}  self-averaging of the quantum free energy has been established in \cite{Craw07}.
	Since we restrict ourselves to the Gaussian case, this property follows immediately from the standard Gaussian concentration inequality.
	We therefore include the short argument for pedagogical reasons. 
	\begin{proposition}[\cite{Craw07}]
		There are some constants $ c, C \in (0,\infty) $ such that for any $ p \in [1,\infty] $ the Gaussian concentration estimate
		\begin{equation}  \mathbb{P}\left(\left| \Phi_N^{p}(\beta,\Gamma)  - \mathbb{E}\left[\Phi_N^{p}(\beta,\Gamma) \right]  \right|  > \frac{t \, \beta }{\sqrt{N} } \right) \leq  C \,  \exp\left(- c t^2\right)
		\end{equation}
		holds for all $ t > 0 $ and all $ N \in \mathbb{N} $ .
	\end{proposition}
	\begin{proof}
		The pressure's variations with respect to the i.i.d.\ standard Gaussian variables $ g_{\pmb{j}} $ is
		$$\displaystyle - \frac{\partial \Phi_N^{p}(\beta,\Gamma)}{\partial g_{\pmb{j}}}   = \frac{\beta}{N^{\frac{p+1}{2}} \, 2^{N} Z(\beta, \Gamma) } \sum_{\pmb{\sigma}}    \sigma_{j_1} \cdots  {\sigma}_{j_p} \, \langle \pmb{\sigma} | e^{-\beta H} | \pmb{\sigma} \rangle . $$
		Here and in the following we use Dirac's bracket notation for matrix elements. Consequently, the Lipschitz constant is bounded by 
		$$\displaystyle \sum_{\pmb{j}}  \left( \frac{\partial \Phi_N^{p}(\beta,\Gamma)}{\partial g_{\pmb{j}}}\right)^2  \leq  \frac{\beta^2}{N}.  $$
		The claim thus follows from the Gaussian concentration inequality for Lipschitz functions. 
	\end{proof}

	In the classical case $ \Gamma =0 $, the free energy of any $ p $-spin interaction is given in terms of Parisi's formula \cite{Pan14}.
	One of its main features is a transition at small enough temperatures to a spin glass regime. 
	At $ p = \infty $ this formula takes a simple form:
	\begin{equation}
	\Phi^{\text{REM}}(\beta) = \lim_{N\to \infty}  \Phi_N^{\infty}(\beta,0) =  \left\{ \begin{array}{l@{\quad}r} \tfrac{1}{2} \beta^2 & \mbox{if } \; \beta \leq \beta_c , \\[1ex]   \tfrac{1}{2} \beta_c^2 + (\beta - \beta_c) \beta_c & \mbox{if } \;   \beta > \beta_c .\end{array} \right.
	\end{equation}
	The non-differentiability at $$ \beta_c := \sqrt{2 \ln 2 }  $$
	reflects a first-order freezing transition into a low-temperature phase characterized by the vanishing of the specific entropy.
	
	Under the addition of a constant transversal field, this freezing transition vanishes for $ \Gamma $ large enough. A first-order phase transition into a quantum paramagnetic phase, characterized by 
	$$ \Phi^{\mathrm{PAR}}(\beta \Gamma)  := \ln \cosh\left(\beta \Gamma\right),$$
	occurs at $ \Gamma_c(\beta)  := \beta^{-1} \arcosh\left( \exp\left(  \Phi^{\text{REM}}(\beta)\right)  \right) $. 
	At $ \beta = \infty $ this connects to the known location $ \Gamma_c(\infty) = \beta_c $ of the quantum phase transition of the ground state~\mbox{\cite{JKrKuMa08,AW15}.} 
	
	\begin{figure}[ht]
		\begin{center}
			\includegraphics[width=\textwidth] {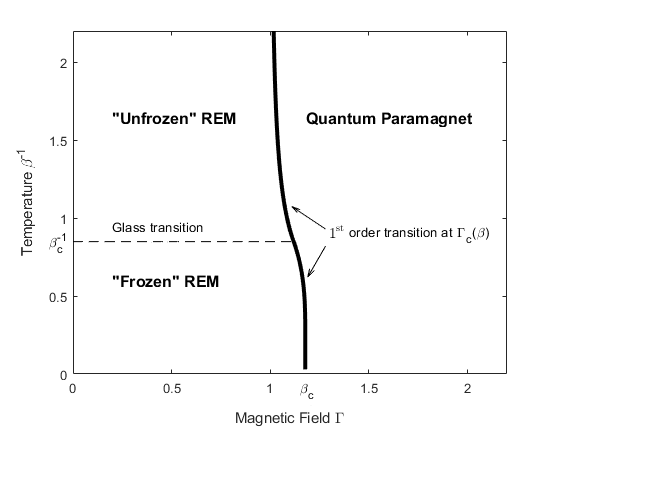}
		\end{center}
		\vspace*{-1cm}
		\caption{Phase diagram of the QREM as a function of the transversal magnetic field $ \Gamma  $ and the temperature $ \beta^{-1}$. The first-order transition occurs at fixed $ \beta $ and $ \Gamma_c(\beta) $. The freezing transition is found at temperature $ \beta_c^{-1} = (2\ln 2)^{-1/2} $, which is unchanged in the presence of a magnetic field of strength $\Gamma < \Gamma_c(\beta)$. } \label{fig:phase}
	\end{figure}
	
	The shape of the phase diagram of the QREM in Figure~\ref{fig:phase} including the precise location of the first-order transition, was predicted by
	Goldschmidt \cite{Gold90} in the 1990s. His arguments are based on the replica trick and the so-called static approximation in the path-integral representation of $\ee[Z_N^p(\beta,\Gamma)^n]$. In a recent paper~\cite{MW19}, we confirmed this prediction. 
	
	\begin{theorem}[\cite{MW19}]\label{thm:Gold}
		For any $ \Gamma , \beta \geq 0 $ almost surely:
		$$
		\Phi^\text{QREM}(\beta,\Gamma) := \lim_{N\to \infty}  \Phi_N^{\textrm{QREM}}(\beta,\Gamma) =\max\{  \Phi^{\text{REM}}(\beta), \Phi^{\text{PAR}}(\beta \Gamma) \}
		. 
		$$
	\end{theorem}

	In broad terms, the main features  of the phase diagram in Figure~\ref{fig:phase} such as a low-temperature frozen phase which gives way to a paramagnetic phase at both high temperatures or strong magnetic field are expected to stay for general $ p $;  cf.~\cite{Gold90,TD90,ONS07,SIC12}. The new features for general $ p $ are the richer structure of the low-temperature phase due to higher-order replica symmetry breaking and the conjectured endpoint of the  first-order  transition line in a critical point at a finite temperature which scales with $ \sqrt{p} $. No closed expression for the free energy is known in the quantum case. Crawford~\cite{Craw07} showed that the almost-sure limit 
	\begin{equation}
	\Phi^p(\beta,\Gamma) := \lim_{N\to \infty} \Phi_N^p(\beta,\Gamma) , 
	\end{equation}
	exists for any $ p \in [1,\infty] $.  All claims concerning the structure of the phase diagram for quantum p-spin models are based  on non-rigorous calculations using the replica trick and a $ 1/p $ expansion \cite{Gold90,TD90,ONS07}. 
	In fact, it is widely believed that $ \Phi^p(\beta,\Gamma) $  
	is continuous in $ 1/p $ and hence tends to the explicit expression for the QREM,
	$$ \lim_{p\to \infty} \Phi^p(\beta,\Gamma)  = \Phi^\textrm{QREM}(\beta,\Gamma) . $$
	We do not quite proof this conjecture in this paper. However, as a main new result we have the following continuity of the free energy.
	\begin{theorem}\label{thm:plimit}
		Let $p(N)$ be a nonnegative sequence which  satisfies a superlogarithmic growth condition, i.e. 
		\begin{equation}\label{eq:growth}
		\lim_{N \to \infty} \frac{p(N)}{\ln(N)} = \infty.
		\end{equation}
		For any $\beta,  \Gamma \geq 0$, we then have the almost sure
		convergence
		\begin{equation}\label{eq:plimit}
		\lim_{N \to \infty} \Phi_N^{p(N)}(\beta,\Gamma) = \Phi^\text{QREM}(\beta,\Gamma).
		\end{equation}
	\end{theorem}
	The proof of this statement heavily relies on the method of proof of Theorem~\ref{thm:Gold} in \cite{MW19}. 
	It will be presented in Section~\ref{sec:proof} below.\\
	
	Let us conclude with some remarks:
	\begin{enumerate}
		\item In the classical case $ \Gamma = 0 $, the quenched pressure $ \mathbb{E}\left[\Phi_N^p(\beta,0) \right] $ is monotonically
		increasing in $ p $ and, in particular, we have
		$  \mathbb{E}\left[\Phi_N^p(\beta,0) \right] \leq  \mathbb{E}\left[\Phi_N^\textrm{REM}(\beta) \right] $ for any $ N \in \mathbb{N} $. 
		This follows with the help of Gaussian comparison \cite[Lemma~10.2.1]{Bov06} 
		from the following facts:
		i)~$ 2 \nn		\ni p \mapsto \mathbb{E}\left[U_p(\pmb{\sigma}) U_p(\pmb{\sigma}') \right]  $ is monotonically decreasing, and
		ii)~$\frac{\partial^2 \Phi_N^p(\beta,0)}{\partial U_p(\pmb{\sigma}) \partial U_p(\pmb{\sigma}') } < 0 $ in case $  \pmb{\sigma} \neq \pmb{\sigma}' $.
			Unfortunately, a similar monotonicity is not known in the quantum case. 
		\item Another intensively studied family of mean-field spin-glasses are the so-called spherical $p$-spin models, given by
		\begin{equation}\label{eq:Utilde}
		\widehat U_p(\pmb{\pmb{\sigma}}) =  \sqrt{\frac{p!}{N^{p-1}} } \sum_{1\leq j_1< \dots < j_p \leq N} g_{j_1,\dots,j_p}{\sigma}_{j_1} \cdots {\sigma}_{j_p}.
		\end{equation}
		In the classical case the spherical $p$-spin  models give rise to the same pressure in the thermodynamic limit as the $p$-spin SK-models \eqref{eq:U}; however, the models have different scales of fluctuations \cite{BKL02}. \\  Theorem \ref{thm:plimit}  remains true (with minor changes in the proof)
		if one works instead with the spherical $p$-spin models.
		This follows from the observation that 
		\begin{equation}\label{eq:covtilde}
		\mathbb{E}\left[\widehat U_p(\pmb{\sigma}) \widehat U_p(\pmb{\sigma}') \right] = N \,  (\xi(\pmb{\sigma},\pmb{\sigma}')^p - \delta_N^p(\pmb{\sigma},\pmb{\sigma^\prime})),
		\end{equation}
		where $\delta_N^p(\pmb{\sigma},\pmb{\sigma^\prime})$ is uniformly  bounded, 
		\begin{equation}\label{eq:delta}
		|\delta_N^p(\pmb{\sigma},\pmb{\sigma^\prime})| \leq \frac{1}{N^p} \sum_{\substack{1\leq j_1\leq \dots \leq j_p \leq N \\ \exists 1\leq k < l \leq p: \, j_k = j_l}} 1\leq \frac{N^p- p! \binom{N}{p}}{N^p} \leq \min \left\{1,\frac{p(p-1)}{2N}\right\}.
		\end{equation}
			 
	\end{enumerate}

	\section{Proof of Theorem~\ref{thm:plimit}}\label{sec:proof} 
	
	The proof is an adaptation of the strategy for the proof Theorem~\ref{thm:Gold}  in \cite{MW19}. 
	The lower bound in \cite{MW19} was based on the Gibbs variational principle and established there already for general p-spin interactions.
	For convenience of the reader, we recall the corresponding lemma here.
	\begin{lemma}[=Lemma 2.1 in \cite{MW19}]\label{lem:lower}
		For any $ p \in [1, \infty] $, $ N \in \mathbb{N} $ and  $ \Gamma , \beta \geq 0 $:
		\begin{equation}
		\Phi_N^p(\beta, \Gamma) \geq  \max\left\{ \Phi_N^p(\beta,0)  , p^{\mathrm{PAR}}(\beta \Gamma) - \frac{\beta }{N 2^N} \sum_{\pmb{\sigma} \in \mathcal{Q}_N} U_p(\pmb{\sigma}) \right\} .
		\end{equation}
	\end{lemma}

	The main new challenge is to cope with the 
	correlations in $ U_p $ for $ p < \infty $ in the upper bound of  \cite{MW19}. 
	Since these correlations vanish in the limit $ p \to \infty $, the large deviation sets
	\begin{equation}\label{eq:deviation}
	\mathcal{L}_\epsilon^p \coloneqq \{\pmb{\sigma}|U_p(\pmb{\sigma}) < -\epsilon N \},
	\end{equation}
	with $\epsilon > 0$  are expected to consist of isolated small clusters. We write $\mathcal{L}_\epsilon^p = \bigcup_\alpha C_\epsilon^{\alpha,p}$ as a disjoint union of as its maximal edge-connected components $ C_\epsilon^{\alpha,p} $, where we recall from~\cite{MW19}:
	\begin{definition}
		An edge-connected component  $   \mathcal{C}_\varepsilon \subset \mathcal{L}_\varepsilon $ is a subset
		for which each pair \linebreak $ \sigma, \sigma' \in  \mathcal{C}_\varepsilon $ is connected through a connected edge-path of adjacent edges. 
		An edge-connected component $  \mathcal{C}_\varepsilon $ is maximal if there is no other vertex $ \sigma \in  \mathcal{L}_\varepsilon \backslash  \mathcal{C}_\varepsilon $ such that $  \mathcal{C}_\varepsilon \cup \{ \sigma \} $ 
		forms an edge-connected component. 
	\end{definition}
	In the situation of Theorem \ref{thm:plimit} we cannot expect that the size of the edge-connected components $C_\epsilon^{\alpha,p(N)}$ remains bounded as $N \to \infty$. However, we show that it is highly likely that all edge-connected components $C_\epsilon^{\alpha,p(N)}$ are contained in balls whose radius grows only sublinearly in~$N$.
	
	\begin{proposition} \label{prop:deviation}
		There exist a subset $\Omega_{\epsilon,N}$ of realisations and a constant $K \in \nn$, which is independent of $N$, such that:
		\begin{enumerate}
			\item[1.] for some $ c_\varepsilon > 0 $, which is independent of $ N $, and all $ N $ large enough:
			$$
			\mathbb{P}\left(  \Omega_{\varepsilon, N} \right) \geq 1 - e^{- c_\varepsilon N } ,
			$$
			\item[2.] on $ \Omega_{\varepsilon, N} $
			any edge-connected component $C_\epsilon^{\alpha,p(N)}$ of $\mathcal{L}_\epsilon^{p(N)}$ is contained in a ball $B_{K \lceil \frac{N}{p(N)} \rceil}(\pmb{\sigma})$ for some $\pmb{\sigma} \in \mathcal{Q}_N$.
	\end{enumerate} \end{proposition}
	
	Before turning to the proof of Proposition~\ref{prop:deviation}, we demonstrate how this result and the basic bounds in  \cite{MW19} imply the almost sure convergence \eqref{eq:plimit} in Theorem~\ref{thm:plimit}.
	
	\begin{proof}[Proof of Theorem \ref{thm:plimit}] 
		The lower bound in Lemma~\ref{lem:lower} yields
		\begin{equation}\label{eq:lbound}
		\begin{split}
		\liminf_{N \to \infty} \Phi_N^{p(N)}(\beta, \Gamma) &\geq \max\{	\liminf_{N \to \infty} \Phi_N^{p(N)}(\beta,0), \Phi^{\text{PAR}}(\beta \Gamma) \} \\ &= \max\{ \ \Phi^{\text{REM}}(\beta), \Phi^{\text{PAR}}(\beta \Gamma) \}
		= \Phi^{\text{QREM}}(\beta,\Gamma),
		\end{split}
		\end{equation}
		Here the last equality follows from the continuity of the classical $p$-spin pressure, which is encoded in Parisi's formula \cite[Thm. 11.3.7]{Bov06}, and its monotonicity, stated as a remark after Theorem 
		\ref{thm:plimit}.
		
		For the upper bound, we fix some $\epsilon >0$ and we use the decomposition of the Hamiltonian 
		\begin{equation}\label{eq:decomp}
		H_{p(N)} \eqqcolon U_{\mathcal{L}_{\epsilon}^{p(N)}} \oplus 
		H_{\mathcal{L}_{\epsilon}^{p(N),c}} - \Gamma A_{\mathcal{L}_{\epsilon}^{p(N)}}
		\end{equation}
		where $ U_{\mathcal{L}_{\epsilon}^{p(N)}}$ is the multiplication operator by the REM values on $\ell^2(\mathcal{L}_{\epsilon}^{p(N)})$ and $H_{\mathcal{L}_{\epsilon}^{p(N),c}}$ is the restriction 
		of the Hamiltonian to the complementary subspace $\ell^2(\mathcal{L}_{\epsilon}^{p(N),c})$. The remainder term 
		$A_{\mathcal{L}_{\epsilon}^{p(N)}}$ consists of the matrix elements of $-T$ reaching $\mathcal{L}_{\epsilon}^{p(N)}$, i.e.
		\begin{equation}\label{eq:matrixel}
		\langle \pmb{\sigma} |  A_{ \mathcal{L}_\varepsilon} | \pmb{\sigma}' \rangle = 
		\begin{cases} 1 & \mbox{if $\pmb{\sigma} \in  \mathcal{L}_\varepsilon $ or $\pmb{\sigma}' \in   \mathcal{L}_\varepsilon $ and $ d(\pmb{\sigma}, \pmb{\sigma}') = 1 $,} \\
		0 & \mbox{else.}
		\end{cases} 	
		\end{equation}
		
		As in the proof of \cite[ Corollary 2.5]{MW19} one obtains from the Golden-Thompson inequality the upper bound 
		\begin{equation}\label{eq:cor25}
		\Phi_N^{p(N)}(\beta, \Gamma) \leq \max\{  \Phi_N^{p(N)}(\beta,0), \Phi^{\text{PAR}}(\beta \Gamma) + \beta \epsilon \} + \frac1N \left(\beta \Gamma \|A_{\mathcal{L}_{\epsilon}^{p(N)}}\| + \ln 2 \right).
		\end{equation}
		
		The operator norm of the restriction of $T$ to a Hamming ball $B_r$ of radius $r$ is known \cite{FriedTill05} to be bounded by 
		$\|T_{B_r}\| \leq 2 \sqrt{r(N-r+1)}$. Since the matrix elements of $-T$ are non-negative,  the restrictions of $T$ satisfy a monotonicity property, i.e. if $A \subset B$, then $\|T_A\| \leq \|T_B\|$. 
		Consequently, on the event $\Omega_{\epsilon,N}$ from Proposition~\ref{prop:deviation} we have
		\begin{equation}\label{eq:normres}
		\limsup_{N \to \infty} \frac1N \|A_{\mathcal{L}_{\epsilon}^{p(N)}}\| = 
		\limsup_{N \to \infty} \max_{\alpha} \frac1N \|T_{\mathcal{C}_{\epsilon}^{\alpha,p(N)}}\| \leq 
		\limsup_{N \to \infty}  \frac{2\sqrt{K}}{\sqrt{N}} \sqrt{1+N/p(N)} = 0 . 
		\end{equation}
		A Borel-Cantelli argument implies the almost sure bound 
		\begin{equation}\label{eq:upbound}
		\limsup_{N \to \infty}\Phi_N^{p(N)}(\beta, \Gamma) \leq
		\Phi^{\text{QREM}}(\beta, \Gamma) + \beta \epsilon,
		\end{equation}
		for any $\epsilon > 0$ and  the assertion of Theorem~\ref{thm:plimit} follows.
	\end{proof}
	
	We prepare the proof of Proposition~\ref{prop:deviation} with a bound on the probability that all components of a centered Gaussian vector are smaller than a certain constant:
	\begin{lemma}\label{lem:gauss}
	 Let $\pmb{g} = (g_1,\ldots,g_L)$, $L\in \nn$ , a centered Gaussian random vector with
		\begin{equation}\label{eq:constant}
		C_L \coloneqq \max_{i = 1,\ldots,L} \sum_{j = 1}^{L} \mathbb{E}\left[g_i g_j\right] . 
		\end{equation}
		Then for any $\delta > 0$ 
		\begin{equation}\label{eq:probbound}
		\pp\left(\max_j g_j  < -\delta,\right)	\leq \exp\left(-\frac{L \delta^2}{2C_L}\right).
		\end{equation}
	\end{lemma}
	\begin{proof}
		The random variable $S_L \coloneqq \sum_{i=1}^L g_i$ is Gaussian, centered and with variance bounded by $\mathbb{E}(S_L^2) \leq LC_L$.
		A standard estimate for Gaussian variables implies 
		\begin{equation}\label{eq:gauss}
		\pp\left(\max_j g_j  < -\delta \right) \leq \pp(S_L < -L \delta)
		\leq \exp\left(- \frac{L^2 \delta^2}{2\mathbb{E}(S_L^2)}
		\right) \leq
		\exp\left(-\frac{L \delta^2}{2C_L}\right). 
		\end{equation}
		\vspace{0.1cm}
	\end{proof}
	
	We are now ready to spell out  the proof of Proposition~\ref{prop:deviation}, which is based on a combinatorial argument. 
\begin{proof}[Proof of Proposition \ref{prop:deviation}]
It turns out to be helpful for the purpose of this proof to introduce the notion of an edge-connected ray. We say that $\pmb{\sigma_1},\ldots,\pmb{\sigma_L}\in \mathcal{Q}_N$ form an \textit{edge-connected ray of length} $L$ if the following properties are satisfied:
\begin{itemize}
	\item $d(\pmb{\sigma_i},\pmb{\sigma_{i+1}}) = 1$ or $d(\pmb{\sigma_i},\pmb{\sigma_{i+1}}) = 2$ for any $i = 1,\ldots,L-1$, \\
	\item $\displaystyle d(\pmb{\sigma_1},\pmb{\sigma_j}) = \sum_{i=1}^{j-1} d(\pmb{\sigma_i},\pmb{\sigma_{i+1}})$ for any $j = 2,\ldots,L,$ 
\end{itemize}
where $d(\pmb{\sigma}, \pmb{\sigma^\prime}) \coloneqq \frac12 \sum_{i=1}^N
|\sigma_i - \sigma_i^\prime|$ denotes the Hamming distance.
Here, the first property ensures that $\pmb{\sigma_1},\ldots,\pmb{\sigma_L}$ form an edge-connected subset of $ \mathcal{Q}_N$ and the second property forces the vertices to form a straight ray starting at $\pmb{\sigma_1}$. \\We now proceed in three steps. In the first step we give a bound for the probability that a certain edge-connected ray is a subset of $\mathcal{L}_\epsilon^{p}$. Then, we consider the probability that $\mathcal{L}_\epsilon^{p}$ contains an edge-connected ray of length $L$. Finally, we use the result from Step 2 to conclude the assertions of Proposition \ref{prop:deviation}.\\

\textbf{Step 1:} Let $\pmb{\sigma_1},\ldots,\pmb{\sigma_L}$ be an edge-connected ray of length $L$. We are interested in the probability that $\{\pmb{\sigma_1},\ldots,\pmb{\sigma_L}\} \subset \mathcal{L}_\epsilon^{p}$. In view of Lemma \ref{lem:gauss}, we calculate 
\begin{align}\label{eq:cov1} 
\sum_{i=1}^L \mathbb{E}[U_p(\pmb{\sigma_i}) U_p(\pmb{\sigma_j})] &= N 
\sum_{i=1}^L \left(1-\frac{2d(\pmb{\sigma_i},\pmb{\sigma_j})}{N}\right)^{p} 
\leq 2N \sum_{k=0}^{L} \left(1-\frac{2k}{N}\right)^{p} \notag \\
&\leq 2N \sum_{k=0}^{L} e^{-2kp/N} \leq \frac{2N}{1-e^{-2p/N}}.
\end{align}
The first equality directly follows from \eqref{eq:spinp} and the next inequality is based on the observation that for any vertex $\pmb{\sigma_i}$ of an edge-connected ray and any number $0 \leq k \leq L$  there are at most two other vertices at distance $k$. Then, we have made use of the convexity of the exponential function and the geometric series formula. 

We note that the function $h(x) \coloneqq \frac{x}{1-e^{-x}}$ is strictly positive and increasing on the interval $(0,1]$. Therefore, we obtain the bound 
\begin{equation}\label{eq:cov2}
\sum_{i=1}^L \mathbb{E}[U_p(\pmb{\sigma_i})U_p(\pmb{\sigma_j})] 
\leq h(1) \frac{N^2}{p}^,
\end{equation}
and Lemma~\ref{lem:gauss} implies
\begin{equation}\label{eq:probbnd}
\pp(\{\pmb{\sigma_1},\ldots,\pmb{\sigma_L}\} \subset \mathcal{L}_\epsilon^{p})
= \pp(\max_{i=1,\ldots,L} U_p(\pmb{\sigma_i}) < -\epsilon N) 
\leq \exp\left(-\frac{L p \epsilon^2 }{2 h(1)}  \right).
\end{equation}

\bigskip
\textbf{Step 2:} We denote by $D(L,N)$ the number of edge-connected rays of length $L$ in $\mathcal{Q}_N$. We claim that 
\begin{equation}\label{eq:combi}
D(L,N)\leq 2^N N^{2L}
\end{equation}
This can be seen as follows: we have $2^N$ choices for the first vertex $\pmb{\sigma}_1$ and at most $N^2$ choices for any subsequent vertex. The bounds \eqref{eq:probbnd} and \eqref{eq:combi} together with the union bound then yield
\begin{equation}\label{eq:probb}
\pp(\{\exists \,\pmb{\sigma_1},\ldots,\pmb{\sigma_L} \in  \mathcal{L}_\epsilon^{p} \text{ forming an edge-connected ray}  \})
\leq  2^N  N^{2L} \exp\left(-\frac{L p \epsilon^2 }{2 h(1)}  \right)
\end{equation}

\bigskip
\textbf{Step 3:} 
We take some fixed $K \in \nn$ and define 
$\Omega_{\epsilon,N}^{K}$ as the subset of realizations where the second assertion holds true. It remains to show the bound  
$\pp((\Omega_{\epsilon,N}^{K})^c) \leq e^{-c_\epsilon N}$ for a convenient choice of $K$.
For any $\omega \notin \Omega_{\epsilon,N}^{K}$, we find an edge-connected component $C_\epsilon^{p(N)}$ of $\mathcal{L}_\epsilon^{p(N)}$ such that  $C_\epsilon^{p(N)} \not \subset B_{K \lceil \frac{N}{p(N)} \rceil}(\pmb{\sigma})$ for any $\pmb{\sigma} \in \mathcal{Q}_N$. In particular, for such an $\omega$ this implies the existence of an edge-connected ray $\pmb{\sigma_1},\ldots,\pmb{\sigma_L} \in \mathcal{L}_\epsilon^{p(N)} $ of length $L \coloneqq \lceil \frac{K}{2} \lceil \frac{N}{p(N)} \rceil \rceil$. Using 
\eqref{eq:probb}, we arrive at
\begin{equation}\label{eq:conclusion}
\begin{split}
\pp((\Omega_{\epsilon,N}^K)^c) &\leq \pp(\{\exists \, \pmb{\sigma_1},\ldots,\pmb{\sigma_L}\in  \mathcal{L}_\epsilon^{p} \text{ forming an edge-connected ray}  \}) \\
& \leq  \exp\left(N\left(2+\frac{K \ln N}{p(N)}-\frac{K \epsilon^2}{4 h(1)}\right) + (K+1) \ln N \right),
\end{split}
\end{equation}
since $2L \leq K (N/p(N)+1) +1$. The first assertion of Proposition \ref{prop:deviation} follows for a suitable choice of $K$, since $p(N)$ satisfies the growth condition \eqref{eq:growth}.
\end{proof}

\section*{Acknowledgements }
This work was partially supported by the DFG under under EXC-2111 -- 390814868.

\bigskip
\bigskip
\begin{minipage}{0.5\linewidth}
\noindent Chokri Manai and Simone Warzel\\
MCQST \& Zentrum Mathematik \\
Technische Universit\"{a}t M\"{u}nchen\\
Corresponding author: \verb+warzel@ma.tum.de+
\end{minipage}

\end{document}